\documentclass[11pt,a4paper]{amsart}
\usepackage[top=4cm,bottom=4cm,left=3.5cm,right=3.5cm]{geometry}

\usepackage{amsmath,amssymb,amsthm}
\usepackage[english]{babel}
\usepackage{mathtools}
\usepackage{enumerate}
\usepackage[hidelinks]{hyperref}
\usepackage{tikz}

\usepackage{xcolor}

\usepackage{algorithm}
\usepackage[noend]{algpseudocode}

\DeclareMathOperator{\Id}{Id}

\newtheorem{theorem}{Theorem}

\newtheorem{lemma}[theorem]{Lemma}
\newtheorem{proposition}[theorem]{Proposition}

\theoremstyle{definition}

\newtheorem{definition}[theorem]{Definition}
\theoremstyle{remark}
\newtheorem{remark}[theorem]{Remark}
\newtheorem{example}[theorem]{Example}

\numberwithin{equation}{section}

\begin{document}

\title[Information Set Decoding in the Lee Metric]{Information Set Decoding in the Lee Metric with Applications to Cryptography}

\author[A.-L. Horlemann-Trautmann]{Anna-Lena Horlemann-Trautmann}
\address{Faculty of Mathematics and Statistics\\
University of St. Gallen\\
Bodanstr. 6\\
St. Gallen, Switzerland\\
}
\email{anna-lena.horlemann@unisg.ch}

\author[V. Weger]{Violetta Weger}
\address{Institute of Mathematics\\
University of Zurich\\
Winterthurerstrasse 190\\
8057 Zurich, Switzerland\\
}
\email{violetta.weger@math.uzh.ch}

\subjclass[2010]{}


\begin{abstract}
We convert Stern's information set decoding (ISD) algorithm to the ring $\mathbb{Z}/4 \mathbb{Z}$ equipped with the Lee metric. Moreover, we set up the general framework for a McEliece and a Niederreiter cryptosystem over this ring. The complexity of the ISD algorithm determines the minimum key size in these cryptosystems for a given security level. We show that using Lee metric codes can substantially decrease the key size, compared to Hamming metric codes. In the end we explain how our results can be generalized to other Galois rings $\mathbb{Z}/p^s\mathbb{Z}$. 
\end{abstract}

\maketitle

\section{Introduction}
\label{sec:introduction}

The hardness of decoding random linear codes is at the heart of any code-based public key cryptosystem.  
Information set decoding (ISD) algorithms are the main method for decoding random linear codes in the Hamming metric, whenever the problem has only a few solutions.  An ISD algorithm is given a corrupted codeword and recovers the message or equivalently finds the error vector.  Such algorithms are often formulated via the parity check matrix, since it is enough to find a vector of a certain weight which has the same syndrome as the corrupted codeword -- this problem is also referred to as the syndrome decoding problem.  ISD algorithms over the binary are based on a decoding algorithm proposed by Prange \cite{prange} in 1962 and the main structure of the variants do not change much from the original: as a first step one chooses an information set, then Gaussian elimination brings the parity check matrix into a standard form and, assuming that the errors are outside of the information set, these row operations on the syndrome will recover the error vector, if the weight does not exceed the given error correction capacity.

ISD algorithms are of immense importance when proposing a code-based cryptosystem. 
The idea of using linear codes in public key cryptography was first formulated by Robert McEliece \cite{mceliece}, in 1978. In the McEliece cryptosystem the private key is the generator matrix of a linear code with an efficient decoding algorithm. The public key is a scrambled and disguised version of the generator matrix, such that the private key (and hence the decoding algorithm) is not reconstructable from the public key. The message is encrypted by encoding it with the generator matrix and adding a random error of prescribed Hamming weight. The owner of the private key can recover the message by inverting the disguising function and using the efficient decoding algorithm. On the other hand, if the secret code is hidden well enough, an adversary who wants to break the system encounters the decoding problem of a random linear code, since the public code looks random to him. The best the adversary can do is hence to use the best generic decoding algorithm for random linear codes, which currently are ISD algorithms. ISD algorithms hence do not break a code-based cryptosystem but they determine the choice of secure parameters.

One of the main drawbacks of classical code-based cryptosystems are the public key sizes. To reduce these key sizes, over the last years, many variants of code-based cryptosystems have been proposed that use codes in the rank-metric, instead of the Hamming metric. This raises the question if other metrics can be useful, as well. This is why we study codes in the Lee metric for code-based cryptography in this paper, and show that, for theoretical code parameters\footnote{By theoretical parameters we mean codes attaining the Gilbert-Varshamov bound.},  the Lee metric can also lead to a substantial reduction of the public key size.

We will focus on codes that are defined over integer residue rings $\mathbb{Z}_m:=\mathbb{Z}/m\mathbb{Z}$, for some integer $m>1$, equipped with the Lee metric. 
In particular, we are going to use ring-linear codes, which are defined to be 
  $\mathbb{Z}_m$-submodules of $\mathbb{Z}_m^n$. We especially focus on quaternary codes, which are defined over $\mathbb{Z}_4$, since this case has been studied the most in the coding theory literature. In general, ring-linear codes were first mentioned by Assmus and Mattson in \cite{ringlin}, for important results see  \cite{blake2, blake1, Ker, nechaev,  saty, shankar,  spiegel}, for a more general overview see \cite{gref}. The idea of using ring-linear codes for cryptography (in a quite different setting) first came up in \cite{ringlincrypto}.

We note that, although $\mathbb Z_4$-linear Lee codes can be represented over $\mathbb F_2$, there exists no representation that preserves both the weight and the linearity of the $\mathbb Z_4$-code over $\mathbb F_2$. Thus the known results over $\mathbb F_2$ cannot be used for the Lee metric. 
This is why this paper presents the adaption of ISD algorithms over the binary \cite{stern} to $\mathbb{Z}_4$ and a general form of the McEliece and Niederreiter cryptosystems over $\mathbb{Z}_4$. The complexity of the ISD algorithm then determines a minimum public key size for a given security level of these cryptosystems. 
The paper is structured as follows: in Section \ref{preliminaries} we introduce the theory of ring-linear codes, especially Lee codes, two ISD algorithms over the binary \cite{stern} and the notations and concepts involved in the algorithms. In Section \ref{ISD over Z4} we present the  adaption of the ISD algorithms over the binary to $\mathbb{Z}_4$, including a complexity analysis. In Section \ref{applications to mceliece over Z4} we cover the applications of the ISD algorithm over $\mathbb{Z}_4$ to code-based cryptography by stating the general McEliece and Niederreiter cryptosystems using quaternary codes. In this context we will also investigate the key size of such a cryptosystem  using theoretical values for the secret  quaternary code regarding   128 bit  security against our ISD algorithms over $\mathbb{Z}_4$, from Section \ref{ISD over Z4}. In Section \ref{generalize} we explain briefly how one can generalize the ISD algorithm as well as the McEliece system to other Galois rings $\mathbb{Z}_{p^s}$.
We will then conclude this paper in Section \ref{conclusion} and add some open questions and problems.

\section{Preliminaries}\label{preliminaries}

In this section we present the main theory and tools of ring-linear codes, especially  Lee codes, as well as some known binary ISD algorithms and the concepts and notations involved. 

\subsection{Ring-linear coding theory}

In traditional finite field coding theory an $[n,k]$ linear code $\mathcal{C}$ over $\mathbb{F}_q$ is a linear subspace of $\mathbb{F}_q^n$ of dimension $k$. One can generalize this by taking a finite ring $R$ instead of $\mathbb{F}_q$. 

Let us assume for simplicity that $R$ is commutative, but observe that the following stays true in the noncommutative case.
\begin{definition}
Let $k$ and $n$ be  positive integers and let $R$ be a finite ring. $\mathcal{C}$ is called an \emph{$R$-linear code of length $n$ of type $h$}, if $\mathcal{C}$ is a submodule of $R^n$, with $\mid \mathcal{C} \mid = h$. 
\end{definition}

We will restrict to the most preferred case of $\mathbb{Z}_m:=\mathbb{Z}/m\mathbb{Z}$, for some $m \in \mathbb{N}$. In particular, we will formulate most of our results for $m=4$, since this case has been studied the most.

\begin{definition} 
We say that $\mathcal{C}$ is a \emph{quaternary linear code} of length $n$, if $\mathcal{C}$ is an additive subgroup of  $\mathbb{Z}_4^n$.
 \end{definition}

In traditional finite field coding theory we endow $\mathbb{F}_q^n$ with the Hamming metric to define  the weight of a codeword $\text{wt}_H$ and the distance of codewords $d_H$. In ring-linear coding theory over $\mathbb{Z}_m$ we could use the Hamming metric, the Lee   metric, the homogeneous metric, the Euclidean metric and so on, for an overview see \cite{gabidulinmetrics}. If we use the Lee metric the corresponding codes are referred to as \emph{Lee codes}. 

\begin{definition}
For $x \in \mathbb Z_m$ we define the \emph{Lee weight} to be
\begin{equation*}
\text{wt}_L(x) = \min\{ x ,  m - x \},
\end{equation*}
similarly for  $x \in \mathbb Z_m^n$ we define the Lee weight to be the sum of the Lee weights of its coordinates:
\begin{equation*}
\text{wt}_L(x) =  \sum_{i=1}^n \text{wt}_L(x_i).
\end{equation*}
%
For $x,y \in \mathbb Z_m^n$, the \emph{Lee distance} is defined to be
\begin{equation*}
d_L(x,y) = \text{wt}_L(x-y).
\end{equation*}
\end{definition}

There is a connection between traditional finite field coding theory and $\mathbb Z_4$-linear coding theory via the Gray map:

\begin{definition}
The Gray map is an isometry between $(\mathbb{Z}_4, \text{wt}_L)$ and $(\mathbb{F}_2^2, \text{wt}_H)$ and  is defined as follows:
\begin{eqnarray*}
\phi: (\mathbb{Z}_4, \text{wt}_L) & \to  & (\mathbb{F}_2^2, \text{wt}_H) \\
0 & \mapsto & (0,0), \\
1 & \mapsto & (0,1), \\
2 & \mapsto & (1,1), \\
3 & \mapsto & (1,0).
\end{eqnarray*}
 The Gray map can be extended componentwise to 
\begin{equation*} 
 \overline{\phi}: (\mathbb{Z}_4^n, \text{wt}_L)  \to   (\mathbb{F}_2^{2n}, \text{wt}_H).
\end{equation*}
\end{definition}
Note however, that the Gray map does not preserve linearity, i.e., the image of a quaternary linear code is generally not linear over $\mathbb F_2$.

 We introduce the following notation: 
For a vector $v$ of length $n$, a matrix $A$ with $n$ columns, a code $\mathcal{C}$ of length $n$ and a set $I \subset \{1, \ldots, n\}$ we denote by $v_I$ the projection of $v$ to its coordinates indexed by $I$, and by $A_I$ the columns of $A$ indexed by $I$. Analogously we define $\mathcal{C}_I := \{ v_I \mid v \in \mathcal{C}\}$. 
 
 We will use the following definition of information set, since it fits perfectly in the context of ring-linear codes:  
 \begin{definition}
 For a code $\mathcal{C}$ over $\mathbb{F}_q$ of length $n$ and dimension $k$, we call a set $I \subseteq \{1, \ldots, n\}$ of size $k$ an \emph{information set} if $\mid \mathcal{C}_I \mid = \mid \mathcal{C} \mid$.
 \end{definition}
Similarly, we define  quaternary information sets for quaternary codes as follows:
 
 \begin{definition}
 For a code $\mathcal{C}$ over $\mathbb{Z}_4$ of length $n$ and type $4^{k_1}2^{k_2}$, we call a set $I \subseteq \{1, \ldots, n\}$ of size $k_1+k_2$ a \emph{(quaternary) information set} if $\mid \mathcal{C}_I \mid = \mid \mathcal{C} \mid$.
 \end{definition}

The following proposition defines the quaternary systematic form of the generator matrix and the parity check matrix of a quaternary code.  

\begin{proposition}[\cite{Ker}]
Let $\mathcal C$ be a quaternary linear code of length $n$ and type $4^{k_1}2^{k_2}$. Then $\mathcal C$ 
 is permutation equivalent to a code having the $(k_1 + k_2)\times n$ generator matrix
\begin{equation}\label{gen}
G = \begin{pmatrix}
\Id_{k_1} & A & B \\
0 & 2\Id_{k_2} & 2C
\end{pmatrix},
\end{equation}
where $A \in \mathbb{Z}_2^{k_1 \times k_2}$,  $B \in \mathbb{Z}_4^{k_1 \times (n-k_1-k_2)}$, $C \in \mathbb{Z}_2^{k_2 \times (n-k_1-k_2)}$, for some $k_1,k_2\in \mathbb N_0$. 

A parity check matrix of $\mathcal{C}$ is the corresponding permutation of the $(n-k_1) \times n$ matrix 
\begin{equation}\label{parity}
H= \begin{pmatrix}
-B^{\top} -C^{\top}A^{\top} & C^{\top} & \Id_{n-k_1-k_2} \\
2A^{\top} & 2\Id_{k_2} & 0
\end{pmatrix} =: \begin{pmatrix}
D & E & \Id_{n-k_1-k_2} \\
2F & 2\Id_{k_2} & 0
\end{pmatrix},
\end{equation}
where $D \in \mathbb{Z}_4^{(n-k_1-k_2) \times k_1}, E  \in \mathbb{Z}_2^{(n-k_1-k_2) \times k_2}, F \in \mathbb{Z}_2^{k_2 \times k_1}$. 
\end{proposition}

If we have a generator matrix of the form \eqref{gen}, to get a unique encoding, the messages need to be of the form  $m = (m_1, m_2)$, where $m_1 \in \mathbb{Z}_4^{k_1}$ and $m_2 \in \mathbb{Z}_2^{k_2}$. 
Encoding is done as follows:
\begin{equation*}
(m_1, m_2) \begin{pmatrix}
\Id_{k_1} & A & B \\
0 & 2\Id_{k_2} & 2C
\end{pmatrix} =\begin{pmatrix}
m_1^\top \\ (m_1A+2m_2)^\top \\ (m_1B+2m_2C)^\top
\end{pmatrix} = \begin{pmatrix}
c_1^\top \\ c_2^\top \\ c_3^\top
\end{pmatrix}.
\end{equation*}
Hence the codewords are of the form $c= (c_1, c_2, c_3)$, where $c_1 \in \mathbb{Z}_4^{k_1}, c_2 \in \mathbb{Z}_4^{k_2}$ and $c_3 \in \mathbb{Z}_4^{n-k_1-k_2}$.

For the syndrome of a codeword $c= (c_1, c_2, c_3)$ we get
\begin{equation*}
\begin{pmatrix}
D & E & \Id_{n-k_1-k_2} \\
2F & 2\Id_{k_2} & 0
\end{pmatrix} \begin{pmatrix}
c_1^\top \\ c_2^\top \\ c_3^\top
\end{pmatrix} = \begin{pmatrix}
Dc_1^\top+Ec_2^\top+c_3^\top \\ 2Fc_1^\top+2c_2^\top 
\end{pmatrix} = \begin{pmatrix}
s_1^\top \\ 2s_2^\top
\end{pmatrix}.
\end{equation*}
The syndromes $s=(s_1, 2s_2)$ are such that $s_1 \in \mathbb{Z}_4^{n-k_1-k_2}$ and $s_2 \in \mathbb{Z}_2^{k_2}$.

To compute the number of vectors in $\mathbb{Z}_4^{n}$ having Lee weight $w$, we have to sum over all choices of   $i$ entries having Lee weight 2, of course only until $\lfloor w/2 \rfloor$. For the rest of the $n-i$ entries we are missing a Lee weight of $w-2i$. We will achieve this with entries of Lee weight 1, where for each of the $w-2i$ entries,   there are two choices: either 1 or 3. We will introduce the following notation for the amount of these vectors:
\begin{equation*}
c(n,w) := \sum_{i=0}^{\lfloor w/2 \rfloor} \binom{n}{i}  \binom{n-i}{w-2i}2^{w-2i}.
\end{equation*}
With the Gray isometry we have that the number of vectors in $\mathbb{Z}_4^{n}$ having Lee weight $w$ is the same as the number of vectors in $\mathbb{F}_2^{2n}$ having Hamming weight $w$, which is simply given by 
$\binom{2n}{w}.$
Note that one can also check that 
\begin{equation}\label{amountlee}
c(n,w) = \sum_{i=0}^{\lfloor w/2 \rfloor} \binom{n}{i}  \binom{n-i}{w-2i}2^{w-2i} = \binom{2n}{w}.
\end{equation}

With this we can easily derive an analogue of the binary Gilbert-Varshamov bound for quaternary codes.

\begin{proposition}[Theorem 13.73, \cite{berlekampbook}]\label{gv}
Let $n$ and $d$ be positive integers. There exists a linear binary code $\mathcal C$  of length $n$ and minimum Hamming distance $d$, such that
\begin{equation*} 
 \mid \mathcal{C} \mid \geq  \frac{2^n }{\sum_{j=0}^{d-1}\binom{n}{j}}.
 \end{equation*}
 Furthermore there exists a linear quaternary code $\mathcal C$  of length $n$ and minimum Lee distance $d$, such that
\begin{equation*} 
 \mid \mathcal{C} \mid \geq  \frac{4^n }{(\sum_{j=0}^{d-1}\binom{2n}{j}-1)3+1}.
 \end{equation*}
\end{proposition}

\subsection{Information set decoding algorithms}\label{subsecstern}

Many ISD algorithms and improvements have been suggested to Prange's simplest form of ISD (see for example \cite{chabanne, canteautsendrier, chabaud, dumer, kruk,  vantilborg}); in historical order the proposed ISD algorithms are by Prange \cite{prange}, Leon \cite{leon}, Lee-Brickell \cite{leebrickell},  Stern \cite{stern}, Canteaut and Chabaud \cite{canteaut}, Finiasz and Sendrier \cite{sendrier},  Bernstein, Lange and Peters \cite{ballcoll}, May, Meurer and Thomae \cite{MMT}, Becker, Joux, May and Meurer \cite{BJMM} and the latest improvement is by May and Ozerov \cite{MO}. 

All of the above mentioned ISD algorithms were proposed over the binary field. However, with new variants of the McEliece cryptosystem proposed over general finite fields, some of the mentioned ISD algorithms have been generalized to $\mathbb{F}_q$: Coffey and Goodman \cite{coffey}  generalized Prange's algorithm to $\mathbb{F}_q$, in \cite{peters} Peters generalized the algorithms by Lee-Brickell and Stern. 
Niebuhr, Persichetti, Cayrel, Bulygin and Buchmann \cite{Niebuhr} generalized the algorithm of Finiasz-Sendrier with efficiency improvements by using partial knowledge of attackers to a general finite field. In \cite{ballcollFq} Interlando, Khathuria, Rohrer, Rosenthal and Weger generalized the ball-collision algorithm by Bernstein, Lange and Peters to $\mathbb{F}_q$.
In \cite{hirose} Hirose generalized the May-Ozerov algorithm to $\mathbb{F}_q$. And Meurer generalized the algorithm of Becker, Joux, May and Meurer in \cite{meurer}. \\

The general idea of ISD algorithms is to guess an information set $I \subset \{1, \ldots n\}$ of size $k$ and the right distribution of the error vector corresponding to this information set, such that we can recover the message from this information set. 
In the algorithms we consider the information set $I$ will be chosen randomly in each outer loop of the algorithm.  Nevertheless we want to note that  there is a slightly smarter way to do so, see \cite{canteaut}, by reusing some elements of $I$ in the next iteration and only adding missing elements. For simplicity, we will just use a random choice.

Once we have chosen an information set $I$, we need to guess the error vector having the assumed weight distribution. In the binary case this means we just have to guess the locations of the errors. 
In Lee-Brickell's algorithm \cite{leebrickell}, the distribution of the error vector is assumed to be $w$ in the information set and $t-w$ outside the information set. 
In Stern's algorithm the error vector has weight $2v$ in the information set coordinates; moreover, these $2v$ errors are  assumed to be located in two disjoint coordinate sets, both having weight $v$. In addition, we assume that there is a zero-window of size $\ell$, where no errors are allowed. The remaining error weight of $t-2v$ is found in the remaining $n-k-\ell$ coordinates.

The average complexity of ISD algorithms  is  given by the cost of one iteration times the average number of iterations needed, which is given by the inverted success probability. Note that the success criterion is to choose the correct weight distribution of the error vector.
The success probability of having correctly chosen $w$ errors in $k$ coordinates over all vectors having length $n$ and Hamming weight $t$ is given by
$$\binom{k}{w} \binom{n}{t}^{-1} .$$

While on classical computers ISD attacks with a high cost of one iteration but a low number of iterations outperform ISD attacks with a low cost of one iteration and a high number of iterations, this is not the case for quantum computers. In fact: in \cite{GroverISD} it was observed that using Grover's algorithm within ISD attacks reduces the number of iterations needed, thus when using a quantum computer ISD attacks with a low cost of one iteration, such as Lee-Brickell's algorithm, might outperform ISD attacks a low number of iterations, such as Stern's algorithm. This is why we will adapt both, Lee-Brickell's algorithm and Stern's algorithm, to the Lee metric.

We will start with explaining the ISD algorithm of Lee-Brickell \cite{leebrickell} over the binary field with respect to the Hamming distance. 
For this we are going to use the following notation. 
 For $S \subset    \{1, \ldots, n \}$ of size $\ell$,  we denote by $\mathbb F_2^n (S)$  the vectors living in $\mathbb F_2^n $ having support in $S$. The projection of $x \in \mathbb F_2^n (S)$ to $\mathbb F_2^{\ell}$ is   denoted by $\pi_S(x)$. On the other hand we denote by $\sigma_S(x)$ the canonical embedding of $x \in \mathbb F_2^{\ell}$ to $ \mathbb F_2^n (S)$.
   We are given the parity check matrix $H \in \mathbb{F}_2^{(n-k) \times n}$, the amount $t$ of errors we can correct and the syndrome  $s \in \mathbb{F}_2^{n-k}$. We want to find a vector $e \in \mathbb{F}_2^n$,   such that $\text{wt}_H(e) = t$ and  $He^{\top} = s$. The algorithm is formulated in Algorithm \ref{leebrickellbinary}.

 \begin{algorithm}[ht!]
\caption{Lee-Brickell's algorithm over $\mathbb{F}_2$}\label{leebrickellbinary}
\begin{flushleft}
Input: The  parity check matrix $H \in \mathbb{F}_2^{(n-k) \times n}$, $s\in\mathbb{F}_2^{n-k}$ and the positive integer $w \in \mathbb{Z}$, such that  $w \leq t$, $w \leq k$ and $t-w \leq n-k$.\\ 
Output: $e\in\mathbb{F}_2^n$ with $He^{\top}=s^\top$ and $wt_H(e)=t$.
\end{flushleft}
\begin{algorithmic}[1]
\State Choose an information set $I\subset \{1, ...,n\}$  of size $k$, let $J= \{1, ...,n\} \setminus I$.
\State Find an invertible matrix $U\in\mathbb{F}_2^{(n-k) \times (n-k)}$ such that $(UH)_{J}=\Id_{n-k}$ and $(UH)_I=A$, where $A \in\mathbb{F}_2^{(n-k) \times k}$.
\State Compute $Us^\top=s'^\top$.
		\For{each $e_1 \in\mathbb{F}_2^n(I)$, with $wt_H(e_1)=w$}
		\If{$wt_H(s'+\pi_I(e_1)A^\top)=t-w$: }
		\State  Output: $e=e_1 +\sigma_{J}(s'+\pi_I(e_1)A^\top)$.
		\EndIf
		\EndFor
		\State  Start over with Step 1 and a new selection of $I$.
\end{algorithmic}
\end{algorithm}

To illustrate the algorithm, assume that the information set is $I=\{1, \ldots, k\}$. We get
\begin{equation*}
UHe^{\top} = \begin{pmatrix}
A  & \Id_{n-k}  
\end{pmatrix} \begin{pmatrix}
e_1^\top   \\ e_2^\top
\end{pmatrix} = s'^\top =Us^\top,
\end{equation*}
where $A \in \mathbb{F}_2^{(n-k) \times k}$  and  $e_1 \in \mathbb{F}_2^k,  e_2 \in \mathbb{F}_2^{n-k}$.
From this we get the condition
\begin{equation*}
 e_1A^\top +e_2 = s'.
\end{equation*}
The part $e_1$ is chosen to have weight $w$, and the part $e_2$ is chosen  such that its support is disjoint from $e_1$, it has the remaining weight $t-w$, and $e_2=s'+e_1 A^\top$.

\begin{remark}
In all the following complexity analyses we will use schoolbook long multiplication with a computational complexity of $n^2$ operations for inputs of length $n$. We remark that faster multiplication algorithms are known and can be used in our ISD algorithms; however, we refrain from using them since they would not make a substantial difference in our analyses, but on the other hand make the formulas more complicated.
\end{remark}

\begin{theorem}
The average number of bit operations Algorithm \ref{leebrickellbinary} needs is approximately
\begin{align*}
& \left( \binom{k}{w} \binom{n-k}{t-w}  \right)^{-1} \binom{n}{t} \cdot \left[  (n-k)^2(n+1) + \binom{k}{w}(w+1)(n-k) \right].
\end{align*}
\end{theorem}
\begin{proof}

As a first step we need to find the  systematic form of the permuted parity check matrix, and the corresponding syndrome form. As a broad estimate we use the complexity of computing $U[H \mid s^\top]$, which takes approximately $(n-k)^2(n+1)$ bit operations.

For all  $e_1 \in\mathbb{F}_2^n(I)$ having $wt_H(e_1)=w$, which are $\binom{k}{w}$ many, we have to check, if the weight of $e_2=s'+\pi_I(e_1)A^\top$ is $t-w$, hence this step costs $(w+1)(n-k)$ bit operations. 

The success probability is given by having chosen the correct weight distribution of the error vector, i.e., in the information set the weight $w$   and in $J$ the missing weight $t-w$:
\begin{equation*}
 \binom{k}{w} \binom{n-k}{t-w}  \binom{n}{t}^{-1}.
\end{equation*}
Hence, the overall cost of this algorithm is as claimed.
\end{proof}

In the following we explain Stern's algorithm over the binary field with respect to the Hamming distance.
  We will use a formulation of the algorithm, which matches the ball-collision formulation in \cite{ballcoll}. The two algorithms differ in the zero-window of size $\ell$: in Stern's algorithm no error is allowed in the zero-window, whereas in the ball-collision algorithm, this window is split into $Y_1, Y_2$ and $q_1, q_2$ errors are allowed respectively. Even though the asymptotic complexity of the ball-collision algorithm is smaller, for concrete parameters it turns out that, in most of the cases, $q_1=q_2=0$ is the most efficient choice, therefore we will generalize Stern's algorithm to $\mathbb{Z}_4$ in this paper.
Nevertheless, using the (ball-)collision formulation allows us to use improvements and speed ups, some of which will be explained in the following, together with their complexities:

\begin{enumerate}
\item
The concept of \emph{intermediate sums} presented in \cite{ballcoll} is important whenever one wants to  do a computation for all vectors in a certain space. Consider, for example, the setting where we are given a binary $k \times n$ matrix $A$ and want to compute $Ax^{\top}$ for all $x \in \mathbb{F}_2^n$, of weight $w$. This would usually cost $k(w-1)$ additions and $w$ multiplications, for each $x \in \mathbb{F}_2^n$. But if we first compute $Ax^{\top}$, where $x$ has weight one, this only outputs the corresponding column of $A$ and has no costs. From there we can compute the sums of two columns of $A$, there are $\binom{n}{2}$ many of these sums and each one costs $k$ additions. From there we can compute all sums of three columns of $A$, which are $\binom{n}{3}$ many. Using the sums of two columns, we only need to add one more column costing $k$ additions. Proceeding in this way, until one reaches the weight $w$, computing $Ax^{\top}$ for all $x \in \mathbb{F}_2^n$ of weight $w$, costs $k  ( L(n,w) -n)$ bit operations, where 
\begin{equation*}
L(n,w) := \sum_{i=1}^w \binom{n}{i}.
\end{equation*}
Note that the we need to take away the cost of the weight one vectors, since they are for free.

\item
The next concept, called \emph{early abort} (also presented in  \cite{ballcoll}), is important whenever a computation is done while checking the weight of the result. For example one wants to compute $x+y$, where $x,y \in \mathbb{F}_2^n$, which usually costs $n$ additions, but we only proceed in the algorithm  if $wt_H(x+y) = t$. Hence we compute and check the weight simultaneously, and if the weight of the partial solution exceeds $t$, we do not need to continue. For the Hamming weight one expects a randomly chosen bit to have weight 1 with probability $\frac{1}{2}$, hence after $2t$ we should reach the wanted weight $t$, and after $2(t+1)$ we should exceed the weight $t$. Hence on average we expect  to compute only $2(t+1)$ many bits of the solution, before we can abort.

\item
The third concept is the idea of using \emph{collisions}. Instead of going through all possible error vectors of weight $2v$, fulfilling certain properties, we can split this process. For this we consider vectors of weight $v$ in one set $S$, and vectors with disjoint weight $v$ in another set $T$, such that two vectors in the intersection of $S$ and $T$ determine the final error vector. The exact definition of the two sets and the properties the vectors need to fulfill will become clearer when we describe the actual algorithm, but we can assume that we need to find 
$x,y$ from some $S\subseteq \mathbb F_2^n$ and $T\subseteq \mathbb F_2^n$, respectively, such that $Ax^{\top} = By^{\top} +s^\top$ for some prescribed $A,B\in \mathbb F_2^{k\times n}$ and $s\in \mathbb F_2^{k}$. 
Assuming that $x,y$ are uniformly distributed, the average amount of collisions is given by 
\begin{equation*}
\mid S \mid \cdot \mid T \mid  \cdot 2^{-n}  .
\end{equation*}
The assumption of a uniform distribution is commonly used, and justified e.g.\ in \cite{ballcoll} and references therein.
\end{enumerate}

The setting for the algorithm is as follows: 
   We are given the parity check matrix $H \in \mathbb{F}_2^{(n-k) \times n}$, the amount $t$ of errors we can correct and the syndrome  $s \in \mathbb{F}_2^{n-k}$. We want to find a vector $e \in \mathbb{F}_2^n$,   such that $\text{wt}_H(e) = t$ and  $He^{\top} = s^\top$. We are going to use all the ideas mentioned above. Stern's algorithm over the binary in the collision formulation is given  in Algorithm \ref{sternbinary}. Note that, without formulating it in detail, the concept of intermediate sums is used in lines 6 and 7, and the concept of early abort is used in line 10. The collision is used in lines 8 and 9, since it is the same $a$ in both cases.

\begin{algorithm}[ht!]
\caption{Collision ISD (Stern's algorithm) over $\mathbb{F}_2$}\label{sternbinary}
\begin{flushleft}
Input: The  parity check matrix $H \in \mathbb{F}_2^{(n-k) \times n}$, $s\in\mathbb{F}_2^{n-k}$ and the positive integers $v, m_1,\, m_2,\, \ell \in \mathbb{Z}$, such that $k=m_1+m_2$, $v\leq m_1, m_2$, $\ell \leq n-k$ and $t-2v\leq n-k-\ell$.\\ 
Output: $e\in\mathbb{F}_2^n$ with $He^{\top}=s^\top$ and $wt_H(e)=t$.
\end{flushleft}
\begin{algorithmic}[1]
\State Choose an information set $I\subset \{1, ...,n\}$  of size $k$.
\State  Choose a set $J \subset \{1, \ldots, n \} \setminus I$ of size $\ell$.
\State Choose a uniform random partition of $I$ into disjoint sets $X$ and $Y$ of size $m_1$ and $m_2=k-m_1$, respectively.
\State Find an invertible matrix $U\in\mathbb{F}_2^{(n-k) \times (n-k)}$ such that $(UH)_{I^C}=\Id_{n-k}$ and $(UH)_I=\begin{pmatrix}
	A_1 \\ A_2
	\end{pmatrix}$, where $A_1\in\mathbb{F}_2^{\ell \times k}$ and $A_2\in\mathbb{F}_2^{(n-k-\ell)\times k}$.
\State Compute $Us^\top= \begin{pmatrix}
	s_1^\top \\ s_2^\top
	\end{pmatrix}$ with $s_1\in\mathbb{F}_2^{\ell}$ and $s_2\in\mathbb{F}_2^{n-k-\ell}$.
\State Compute the set $S$ consisting of all pairs $(\pi_I(e_X)A_1^\top, e_X)$, where $e_X \in\mathbb{F}_2^n(X)$, $wt_H(e_X)=v$.	
\State Compute the set $T$ consisting of all pairs $(\pi_I(e_Y)A_1^\top+s_1,e_Y)$, where $e_Y \in\mathbb{F}_2^n(Y)$, $wt_H(e_Y)=v$.
\For{each $(a, e_X)\in S$}
		\For{each $(a, e_Y) \in T$}\label{collstep}
		\If{$wt_H(\pi_I(e_X + e_Y)A_2^\top+s_2)=t-2v$: }\label{stepweight}
		\State  Output: $e=e_X + e_Y +\sigma_{J}(\pi_I(e_X + e_Y)A_2^\top+s_2)$.
		\EndIf
		\EndFor
		\EndFor
		\State  Start over with Step 1 and a new selection of $I$.
\end{algorithmic}
\end{algorithm}

Let us illustrate the algorithm in the easiest situation, where the information set is $I=\{1, \ldots, k\}$ and the zero window is $\{k+1, \ldots, k + \ell\}$. We get
\begin{equation*}
UHe^{\top} = \begin{pmatrix}
A_1 & \Id_{\ell} & 0 \\
A_2 & 0 & \Id_{n-k-\ell}
\end{pmatrix} \begin{pmatrix}
e_1^\top \\ 0 \\ e_2^\top
\end{pmatrix} = \begin{pmatrix}
s_1^\top \\ s_2^\top
\end{pmatrix} =Us^\top,
\end{equation*}
where $A_1 \in \mathbb{F}_2^{\ell \times k}$, $A_2 \in \mathbb{F}_2^{(n-k-\ell) \times k}$ and  $e_1 \in \mathbb{F}_2^k,  e_2 \in \mathbb{F}_2^{n-k-\ell}, s_1 \in \mathbb{F}_2^{\ell}, s_2 \in \mathbb{F}_2^{n-k-\ell}$.
From this we get the conditions 
\begin{eqnarray*}
 e_1A_1^\top  &=& s_1, \\
e_1A_2^\top + e_2 &=& s_2.
\end{eqnarray*}
The part $e_1$ is chosen to be $ \pi_I(e_X + e_Y)$ such that it has weight $2v$, and with the collision we ensure that the first condition is satisfied. The part $e_2$ is chosen  such that its support is disjoint from $e_1$, it has the remaining weight $t-2v$, and the second condition is satisfied, i.e., $e_2=\pi_I(e_X+e_Y)A_2^\top +s_2$.
Therefore 
\begin{align*}
& UHe^{\top} = \begin{pmatrix}
A_1 & \Id_{\ell} & 0 \\
A_2 & 0 & \Id_{n-k-\ell}
\end{pmatrix} \begin{pmatrix}
\pi_I(e_X+e_Y)^\top \\ 0 \\ A_2\pi_I(e_X + e_Y)^\top +s_2^\top
\end{pmatrix}  \\
& =\begin{pmatrix}
A_1 \pi_I(e_X+ e_Y)^\top \\ A_2\pi_I(e_X + e_Y)^\top+A_2\pi_I(e_X+ e_Y)^\top +s_2^\top
\end{pmatrix}
=
 \begin{pmatrix}
s_1^\top \\ s_2^\top
\end{pmatrix} =Us^\top,
\end{align*}
i.e., $e=(e_1,0,e_2)$ fulfills $He^{\top}=s^\top$ and $wt_H(e)=t$. \\

\begin{theorem}
The average number of bit operations Algorithm \ref{sternbinary} needs, is  
\begin{align*}
& \left( \binom{m_1}{v} \binom{m_2}{v} \binom{n-k -\ell}{t-2v} \right)^{-1} \binom{n}{t} \cdot \left[  (n-k)^2(n+1) + \ell (L(m_1,v)-m_1) \right.  \\
& \left. +  \ell \left( L(m_2,v)-m_2 + \binom{m_2}{v}\right) +   \binom{m_1}{v}\binom{m_2}{v}  2^{-\ell+1} (t-2v+1)(2v+1) \right].
\end{align*}
\end{theorem}
\begin{proof}
\begin{enumerate}
 \item
As a first step we need to find the  systematic form of the permuted parity check matrix, and the corresponding syndrome form. As a broad estimate we use the complexity of computing $U[H \mid s^\top]$, which takes approximately $(n-k)^2(n+1)$ bit operations.

\item
 To build the set $S$ one has to compute $\pi_{I}(e_X)A_1^\top$  for all $e_X \in \mathbb{F}_2^n(X)$ of weight $v$. Using intermediate sums this costs 
 \begin{equation*}
\ell (L(m_1,v)-m_1).
 \end{equation*}
 
 \item
The set $T$  is built similarly, since one has to compute  $\pi_{I}(e_Y)A_1^\top+s_1$  for all $e_Y \in \mathbb{F}_2^n(Y)$ of weight $v$.  Using intermediate sums this costs
  \begin{equation*}
   \ell \left( L(m_2,v)-m_2+ \binom{m_2}{v}\right),
 \end{equation*}
where $\ell\binom{m_2}{v}$ is the complexity of adding $s_1$ to $\pi_{I}(e_Y)A_1^\top$ for each $e_Y \in \mathbb{F}_2^n(Y)$ of weight $v$.
 \item
 In the next step we want to check for  collisions between the set $S$ and $T$. The set $S$  consists of all $e_X$, where $e_X \in \mathbb{F}_2^n(X)$ has weight $v$. Hence $S$ is of size $\binom{m_1}{v}$ and similarly the set $T$ is of size $\binom{m_2}{v}$. The  collision lives in $\mathbb{F}_2^{\ell}$, hence if we assume an uniform distribution we have to check on average 
 \begin{equation*}
 \frac{\binom{m_1}{v}\binom{m_2}{v}}{2^{\ell}}
 \end{equation*} many collisions. 
For each collision we have to compute $\pi_I(e_X + e_Y)A_2^\top+s_2$ and we only proceed if the weight of this is $t-2v$. With the method of early abort we only have to compute on average $2(t-2v+1)$ many entries. Each entry of the solution costs $2v+1$ bit operations. 

\item
This sums up to the cost of one iteration being

\begin{align*}
\ \quad \quad c(n, k, t, v, m_1, m_2, \ell) = &  (n-k)^2(n+1) \\ 
 & + \ell (L(m_1,v)-m_1)  +  \ell \left( L(m_2,v)-m_2+ \binom{m_2}{v}\right) \\
 & + \frac{\binom{m_1}{v}\binom{m_2}{v}}{2^{\ell}}2 (t-2v+1)( 2v+1).
\end{align*}

The success probability is given by having chosen the correct weight distribution of the error vector, i.e. in $X$ the weight $v$, in $Y$ the weight $v$,  and in $J$ the missing weight $t-2v$:
\begin{equation*}
s(n,k,t,v, m_1, m_2, \ell) = \binom{m_1}{v}\binom{m_2}{v}\binom{n-k -\ell}{t-2v}\binom{n}{t}^{-1}.
\end{equation*}
Hence the overall cost of this algorithm is given  as in the claim by  
 $$ c(n, k, t, v, m_1, m_2, \ell) \cdot s(n,k,t,v, m_1, m_2, \ell)^{-1}.$$
  \end{enumerate}

\end{proof}

  \section{Information set decoding over $\mathbb{Z}_4$}\label{ISD over Z4}
  
 In this section we adapt our previous formulation of Lee-Brickell's and Stern's algorithm (in the collision formulation) for $\mathbb Z_4$-linear codes. For both algorithms, we first formulate the algorithm and illustrate how and why they work, before we determine their complexities. 
 
 Before explaining the algorithms, we determine the complexity of computing $Ax^{\top}$, for $x \in \mathbb{Z}_4^n$ having Lee weight $w$ and a given matrix $A \in \mathbb{Z}_4^{k \times n}$, since this will be used in both algorithms.
 \begin{lemma}\label{lem13}
 Let $A \in \mathbb{Z}_4^{k \times n}$ and $x \in \mathbb{Z}_4^n$ with $wt_L(x)=w$. Then computing
 $Ax^{\top}$ needs at most $2(w-1)k$ bit operations.
  \end{lemma}
 \begin{proof}
 Whenever an entry of $x$ is a $1$, we need to add a column, and whenever it is a $3$ we need to subtract a column. Both, a column addition or subtraction, cost at most $2k$ bit operations. If the entry of $x$ is $2$, we have to add a column twice; but the entry also has Lee weight $2$, i.e., per Lee weight one we still get $2k$ operations. Thus, the overall complexity is $2(w-1)k$.
 \end{proof}
 
  \subsection{Lee-Brickell's algorithm over $\mathbb{Z}_4$} 
 
Recall that the systematic form of the parity check matrix is permutation equivalent to \eqref{parity}. 
For our purpose however, it is enough  to consider the systematic form as
 \begin{equation*}
 \begin{pmatrix}
A & \Id_{n-k_1-k_2} \\
 2C &   0
 \end{pmatrix},
 \end{equation*}
 where $A \in \mathbb{Z}_4^{(n-k_1-k_2)\times (k_1+k_2)}$ and $C \in \mathbb{Z}_2^{k_2 \times  (k_1+k_2)} $.  The algorithm is given in Algorithm \ref{leebrickellnew}.
 
 To illustrate the algorithm, let us again assume that the chosen information set is $\{1, \ldots, k\}$. Then there exists an invertible $U\in\mathbb{Z}_4^{(n-k_1)\times(n-k_1)} $, such that
  \begin{equation*}
 UHe^{\top} = \begin{pmatrix}
 A &  \Id_{n-k_1-k_2}  \\
 2C  & 0  
 \end{pmatrix}\begin{pmatrix}
 e_1^\top \\ e_2^\top
 \end{pmatrix} = \begin{pmatrix}
 s_1^\top \\ 2s_2^\top  
 \end{pmatrix} =Us^\top,
 \end{equation*}
 where $s_1 \in \mathbb{Z}_4^{n-k_1-k_2}, s_2 \in \mathbb{Z}_2^{k_2}$  and $e_1 \in \mathbb{Z}_4^{k_1+k_2}, e_2  \in \mathbb{Z}_4^{n-k_1-k_2}$. From this we get the conditions
 \begin{align*}
 e_1A^\top +e_2= & s_1, \\
 2e_1C^\top = & 2s_2. 
 \end{align*}
  We will choose $e_1$ having Lee weight $w$ and such that the second condition is satisfied. The first condition will then be satisfied by the choice of $e_2$ and we check that $e_2$ has the remaining Lee weight $t-w$.
  
\begin{algorithm}[h!]
\caption{Lee-Brickell's Algorithm over $\mathbb{Z}_4$}\label{leebrickellnew}
\begin{flushleft}
Input: The $(n-k_1)\times n$ parity check matrix $H$ over $\mathbb{Z}_4$, the syndrome $s\in\mathbb{Z}_4^{n-k_1}$ and the positive integers $t,w  \in \mathbb{Z}$, such that  $w \leq 2(k_1+k_2)$, $w  \leq t$ and  $t-w\leq 2(n-k_1-k_2)$.\\ 
Output: $e\in\mathbb{Z}_4^n$ with $He^{\top}=s^\top$ and $wt_L(e)=t$.
\end{flushleft}
\begin{algorithmic}[1]
\State Choose a quaternary information set $I \subset \{1, \ldots, n\}$ of size $k_1+k_2$, let $J=  \{1, \ldots, n\} \setminus I$.
\State  Find an invertible matrix $U\in\mathbb{Z}_4^{(n-k_1)\times(n-k_1)} $, such that  \begin{equation*}
(UH)_{I} = \begin{pmatrix}
A \\  2C
\end{pmatrix},  \quad (UH)_{J}  = \begin{pmatrix}
 \Id_{n-k_1-k_2} \\ 0
\end{pmatrix}, 
\end{equation*}
  where $A \in \mathbb{Z}_4^{(n-k_1-k_2)\times (k_1+k_2)}$ and $C \in \mathbb{Z}_2^{k_2 \times  (k_1+k_2)} $.
\State Compute $Us^\top= \begin{pmatrix} s_1^\top \\ 2s_2^\top \end{pmatrix}$, where  $s_1 \in \mathbb{Z}_4^{n-k_1-k_2}, s_2 \in \mathbb{Z}_2^{k_2}$ .
	\For{$e_1 \in \mathbb{Z}_4^n(I), wt_L(e_1)=w$}
	\If{$2\pi_I(e_1)C^\top=2s_2$}
	\If{$wt_L( s_1- \pi_{I}(e_1)A^\top)=t-w$}
	\State  Output: $e= e_1+ \sigma_{J}(s_1- \pi_{I}(e_1)A^\top)$
				\EndIf
		\EndIf
		\EndFor
		\State  	Start over with Step 1 and a new selection of $I$.
\end{algorithmic}
 
\end{algorithm}

  \subsection{Complexity analysis of Lee-Brickell's algorithm over $\mathbb{Z}_4$}

We now estimate the complexity of  Lee-Brickell's algorithm over $\mathbb{Z}_4$. We assume that one addition or one multiplication over $\mathbb Z_4$ costs $2$ binary operations each. 
We remark here that if  a lookup table for the multiplication and addition is used, the cost of one multiplication as well as the cost of one addition over $\mathbb{Z}_4$  will be only one bit. All the following costs, however, will be given without using a lookup table.  
Moreover, as in the binary case, we use school-book long multiplication instead of faster multiplication algorithms.

\begin{theorem}
The average number of bit operations Algorithm \ref{leebrickellnew} needs, is  

$$ \frac{ \binom{2n}{t} }{ \binom{2(k_1+k_2)}{w}   \binom{2(n-k_1-k_2 )}{t-w} } \left[ 2(n-k_1)^2(n+1) + \binom{2(k_1+k_2)}{w}2(w(n-k_1)-k_2)  \right].$$
\end{theorem}
 
 \begin{proof}
 
As a first step we need to find the (permuted) quaternary systematic form of the parity check matrix, and the corresponding syndrome form. As a broad estimate we use the complexity of computing $U[H \mid s^\top]$, which takes approximately
$(n-k_1)^2(n+1)$ quaternary operations, i.e.,
 $2(n-k_1)^2(n+1)$ bit operations.

As a next step, we compute $2\pi_I(e_1)C^\top$ for all $e_1 \in \mathbb{Z}_4^n(I)$ having $wt_L(e_1)=w$. 
With Lemma \ref{lem13} this costs $2(w-1)k_2$ bit operations for one choice of $e_1$.  
Analogously, we get that computing $\pi_{I}(e_1)A^\top$ costs $2(w-1)(n-k_1-k_2)$ bit operations. Thus, computing $s_1-\pi_{I}(e_1)A^\top$ costs $2w(n-k_1-k_2)$ bit operations.
Since there are $\binom{2(k_1+k_2)}{w}$ many such $e_1$ we get an overall cost of
$$\binom{2(k_1+k_2)}{w}(2w(n-k_1)-2k_2)  .$$

The success probability is given by having chosen the correct weight distribution of the error vector, i.e., weight $t$ in the information set,  and the missing weight $t-w$  in $J$:
\begin{equation*}
  \binom{2(k_1+k_2)}{w} \binom{2(n-k_1-k_2)}{t-w}  \binom{2n}{t}^{-1}.
\end{equation*}
It follows that the overall cost of the algorithm is as claimed.
  \end{proof}

    \subsection{Collision ISD (Stern's algorithm) over $\mathbb{Z}_4$}

As in Lee-Brickell's algorithm, we first bring the parity check matrix into systematic form, according to the chosen information set. Because of the zero window of length $\ell$ we now split the matrix into three block rows, instead of two. If we assume that the  information set is $I=\{1,\dots,k_1+k_2\}$ and the zero window is $\{k_1+k_2+1, \ldots, k_1+k_2+\ell\}$, we get the following situation:
  \begin{equation*}
 UHe^{\top} = \begin{pmatrix}
 A &  \Id_{\ell} & 0 \\
 B & 0 & \Id_{n-k_1-k_2-\ell} \\
 2C  & 0 & 0
 \end{pmatrix}\begin{pmatrix}
 e_1^\top \\ 0  \\ e_2^\top
 \end{pmatrix} = \begin{pmatrix}
 s_1^\top \\ s_2^\top \\ 2s_3^\top
 \end{pmatrix} =Us^\top,
 \end{equation*}
 where $s_1 \in \mathbb{Z}_4^{\ell}, s_2 \in \mathbb{Z}_4^{n-k_1-k_2-\ell}, s_3 \in \mathbb{Z}_2^{k_2}$ and  $A \in \mathbb{Z}_4^{\ell \times (k_1+k_2)}$, \\ $B \in \mathbb{Z}_4^{(n-k_1-k_2-\ell) \times (k_1+k_2)}$,  $C \in \mathbb{Z}_2^{k_2 \times  (k_1+k_2)}$ and $e_1 \in \mathbb{Z}_4^{k_1+k_2}, e_2  \in \mathbb{Z}_4^{n-k_1-k_2-\ell}$. From this we get the conditions
 \begin{align*}
 e_1A^\top= & s_1, \\
 e_1B^\top +e_2 = & s_2, \\
 2e_1C^\top = & 2s_3. 
 \end{align*}
  We will choose $e_1$ and $e_2$ having disjoint Lee weight $2v$ and $t-2v$, respectively. In order to satisfy the first and the third condition, which only depend on $e_1$, we will check for a collision in the algorithm. The second condition will be satisfied by the choice of  $e_2$.
  Thus, compared to the binary version we only get the extra conditions $2e_1C^\top=2s_3$ on $e_1$. The rest is analogous. 
In fact we choose 
$e_1 = \pi_{I}(e_X + e_Y)$ and $e_2 =   s_2- e_1B^\top=s_2-\pi_{I}(e_X + e_Y)B^\top$, where $I$ is the quaternary information set, and  $X$ and $Y$ are partitions of $I$.  Therefore we get
  \begin{equation*}
 UHe^{\top}   =\begin{pmatrix}
 A\pi_{I}(e_X + e_Y)^\top \\
 B \pi_{I}(e_X + e_Y)^\top +s_2^\top-B\pi_{I}(e_X + e_Y)^\top \\
 2C\pi_{I}(e_X + e_Y)^\top 
\end{pmatrix}  =  \begin{pmatrix}
 s_1^\top \\ s_2^\top \\ 2s_3^\top
 \end{pmatrix} =Us^\top.
 \end{equation*}
 
 The final collision algorithm is formulated in Algorithm \ref{sternnew}. As in the binary case, we implicitly assume that we use intermediate sums in lines 6 and 7, early abort in line 10 and  the speed up by using collisions in lines 8 and 9.

\begin{algorithm}[h!]
\caption{Collision ISD (Stern's algorithm) over $\mathbb{Z}_4$}\label{sternnew}
\begin{flushleft}
Input: The $(n-k_1)\times n$ parity check matrix $H$ over $\mathbb{Z}_4$, the syndrome $s\in\mathbb{Z}_4^{n-k_1}$ and the positive integers $v, m_1, m_2,  \ell \in \mathbb{Z}$, such that $k_1+k_2 = m_1 + m_2$, $v \leq 2m_1$, $v \leq 2m_2$,  $2v\leq t$ and $t-2v \leq 2(n-k_1-k_2- \ell)$.\\ 
Output: $e\in\mathbb{Z}_4^n$ with $He^{\top}=s^\top$ and $wt_L(e)=t$.
\end{flushleft}
\begin{algorithmic}[1]
\State Choose a quaternary information set $I \subset \{1, \ldots, n\}$ of size $k_1+k_2$.
\State Choose a set $Z \subset \{1, \ldots, n\} \setminus I$, of size $\ell$ and define $J =\{1, \ldots ,n\} \setminus (I \cup Z)$. 
\State Partition $I$ into two disjoint sets $X$ and $Y$  of size $m_1$ and $m_2= k_1+k_2-m_1$ respectively.
\State\label{find Unew}  Find an invertible matrix $U\in\mathbb{Z}_4^{(n-k_1)\times(n-k_1)} $, such that  \begin{equation*}
(UH)_{I} = \begin{pmatrix}
A \\ B \\ 2C
\end{pmatrix}, \quad (UH)_{Z} 
  = \begin{pmatrix}
\Id_{\ell} \\0\\ 0
\end{pmatrix}, \quad (UH)_{J}  = \begin{pmatrix}
0 \\ \Id_{n-k_1-k_2-\ell} \\ 0
\end{pmatrix}, 
\end{equation*}
 where $A \in \mathbb{Z}_4^{\ell \times (k_1+k_2)}, B \in \mathbb{Z}_4^{(n-k_1-k_2-\ell) \times (k_1+k_2)}, C_1 \in \mathbb{Z}_2^{k_2 \times  (k_1+k_2)}$.
\State Compute $Us^\top= \begin{pmatrix} s_1^\top \\ s_2^\top \\ 2s_3^\top  \end{pmatrix}$, where  $s_1 \in \mathbb{Z}_4^{\ell}, s_2 \in \mathbb{Z}_4^{n-k_1-k_2-\ell}$ and $s_3 \in \mathbb{Z}_2^{k_2}$.
\State Compute the following set \label{buildS}
	\begin{equation*}
	S= \{(\pi_{I}(e_X)A^\top,2\pi_{I}(e_X)C^\top, e_X)  | \  e_X \in \mathbb{Z}_4^n(X), wt_L(e_X) = v \}.
	\end{equation*}
\State Compute the following set \label{buildT}
\begin{equation*}
T=\{(s_1-\pi_{I}(e_Y)A^\top,  2s_3-2\pi_{I}(e_Y)C^\top, e_Y) | \   e_Y \in \mathbb{Z}_4^n(Y), wt_L(e_Y)=v \}.
\end{equation*}
	\For{$(a,  b, e_X)\in S$}
	\For{$(a, b, e_Y)\in T$} \label{collTnew}
	\If{$wt_L( s_2- \pi_{I}(e_X+e_Y)B^\top)=t-2v$}
	\State  Output: $e= e_X+e_Y+ \sigma_{J}(s_2-\pi_{I}(e_X+e_Y)B^\top)$
				\EndIf
		\EndFor
		\EndFor
		\State  	Start over with Step 1 and a new selection of $I$.
\end{algorithmic}
 
\end{algorithm}

\subsection{Complexity analysis of collision ISD over $\mathbb{Z}_4$}

First we determine the complexities of the separate speed up concepts used in the main part of the algorithm. 
 
  \begin{enumerate}
\item
Intermediate sums: The concept is the same over $\mathbb Z_4$ as over $\mathbb F_2$. 
With
\begin{equation*}
\bar{L}(n,w) := \sum_{i=1}^w c(n,i) = \sum_{i=1}^w \binom{2n}{i}
\end{equation*}
and Lemma \ref{lem13} we get that the cost of computing $Ax^{\top}$, for all $x \in \mathbb{Z}_4^n$ of Lee weight $w$, is $2k( \bar{L}(n,w)-2n)$ bit operations.

\item 
Early abort: This concept changes slightly when using the Lee weight over $\mathbb{Z}_4$. On average we can expect $1/2$ of the entries to have weight 1, $1/4$ of the entries to have weight $2$ and $1/4$ of the entries to have weight 0. Hence, on average, we should reach Lee weight $t$ after $(\frac{1}{2}+2 \frac{1}{4})^{-1}t =t$ additions and therefore calculate $t+1$ entries, before we can abort the computation.

\item
Collisions: The average amount of collisions one needs to check between elements living in $\mathbb{Z}_4^n$ of a set $S$ and a set $T$, under the assumption of a uniform distribution (analogously to the binary case), is given by 
$$ \mid S \mid \cdot \mid T \mid \cdot 4^{-n}. $$
\end{enumerate}

\begin{theorem}
The average number of bit operations Algorithm \ref{sternnew} needs, is 
\begin{align*}
& \frac{\binom{2n}{t}}{ \binom{2m_1}{v} \binom{2m_2}{v} \binom{2(n-k_1-k_2-\ell)}{t-2v} }  \left[ 2(n-k_1)^2(n+1) \right. \\
 &  + 2 \ell \left(\bar{L}(m_1,v) +\bar{L}(m_2,v)-2m_1-2m_2+\binom{2m_2}{v}\right) \\
& +k_2\left(L(m_1,v)+L(m_2,v)-m_1-m_2+2+\binom{2m_2}{v}\right)  \\
&  \left.  + \frac{c(m_1,v)c(m_2,v)}{2^{k_2+2\ell}} (t-2v+1)(4v-2)  \right],
\end{align*}
assuming that $v\geq 1$.
\end{theorem}
 
 \begin{proof}
 \begin{enumerate}
 \item
As a first step we need to find the (permuted) quaternary systematic form of the parity check matrix, and the corresponding syndrome form. As a broad estimate we use the complexity of computing $U[H \mid s^\top]$, which takes approximately
$(n-k_1)^2(n+1)$ quaternary operations, i.e.,
 $2(n-k_1)^2(n+1)$ bit operations.

\item
 To build the set $S$ one has to compute $\pi_{I}(e_X)A^\top$ and $2\pi_{I}(e_X)C^\top$ for all \\
 $e_X \in \mathbb{Z}_4^n(X)$ of Lee weight $v$. Using intermediate sums the former costs  $2\ell (\bar{L}(m_1,v)-2m_1)$ and the latter costs $k_2 (L(m_1,v)-m_1+1)$, since $C$ lives in $\mathbb{Z}_2^{k_2 \times (k_1+k_2)}$ and $2\pi_{I}(e_X)$ lives in $2\mathbb{Z}_2^{k_1+k_2}$.  Hence, computing $S$ costs in total 
  \begin{equation*}
2\ell (\bar{L}(m_1,v)-2m_1) +  k_2 (L(m_1,v)-m_1+1) 
 \end{equation*}
binary operations.
 
 \item
The set $T$  is built similarly, since one has to compute  $s_1-\pi_{I}(e_Y)A^\top$ and $2s_3-2\pi_{I}(e_Y)C^\top$ for all $e_Y \in \mathbb{Z}_4^n(Y)$ of weight $v$.  Using intermediate sums 
we get a complexity of
\begin{equation*}
  2 \ell \left(\bar{L}(m_2,v)-2m_2 +\binom{2m_2}{v}\right) +k_2\left( L(m_2,v)-m_2+1 + \binom{m_2}{v}\right).
 \end{equation*}
 binary operations.

 \item
 In the next step we want to check for the \textit{two} collisions between the set $S$ and $T$. The set $S$  consists of all $e_X$, where $e_X \in \mathbb{Z}_4^n(X)$ has weight $v$. Hence $S$ is of size $\binom{2m_1}{v}$ and similarly the set $T$ is of size $\binom{2m_2}{v}$. The first collision lives in $\mathbb{Z}_4^{\ell}$, whereas the second collision lives in $2\mathbb{Z}_2^{k_2}$. We assume an uniform distribution and hence have to check on average 
 \begin{equation*}
 \frac{\binom{2m_1}{v}\binom{2m_2}{v}}{2^{k_2+2\ell}}
 \end{equation*} many collisions. 
For each collision we have to compute $s_2-\pi_{I}(e_X+e_Y)B^\top$. With the method of early abort we only have to compute on average $t-2v+1$ entries. By Lemma \ref{lem13}, each entry of the solution costs $4v-2$ binary operations (assuming that $v\geq 1$). 

\item
This sums up to the cost of one iteration being
 \begin{align*}
\quad \quad c(n, k_1, k_2, t, & m_1, m_2, v, \ell) :=  2(n-k_1)^2(n+1) \\
& + 2 \ell (\bar{L}(m_1,v) -2m_1+\bar{L}(m_2,v)-2m_2) \\&
 +2\ell\binom{2m_2}{v}  +k_2\binom{2m_2}{v}  + k_2(L(m_1,v)-m_1+1)\\
 &+k_2(L(m_2,v)-m_2+1)  + \frac{\binom{2m_1}{v}\binom{2m_2}{v}}{2^{k_2+2\ell}} (t-2v+1)(4v-2).
\end{align*}
The success probability is given by having chosen the correct weight distribution of the error vector, i.e. in $X$ the weight $v$, in $Y$ the weight $v$,  and in $J$ the missing weight $t-2v$:
\begin{equation*}
\quad \quad s(n, k_1, k_2, t, m_1, m_2, v, \ell) := \binom{2m_1}{v} \binom{2m_2}{v} \binom{2(n-k_1-k_2-\ell)}{t-2v}  \binom{2n}{t}^{-1}.
\end{equation*}
Hence the overall cost of this algorithm is as in the claim given by  
\begin{align*}
c(n, k_1, k_2, t, m_1, m_2, v, \ell) \cdot s(n, k_1, k_2, t, m_1, m_2, v, \ell)^{-1}.
\end{align*}
  \end{enumerate}
  \end{proof}

\section{Applications: code-based cryptosystems over $\mathbb{Z}_4$}\label{applications to mceliece over Z4}

In this section we state a quaternary version of the McEliece and the Niederreiter cryptosystem. 
For the key generation one chooses a quaternary code $\mathcal{C}$ of length $n$ and type $h=4^{k_1}2^{k_2}$, which has an efficient decoding algorithm and is able to correct up to $t$ errors. We do not propose the use of a specific code, but we note that the secret code needs to come from a family of codes that is large enough and have a large enough error correction capacity $t$, such that brute force attacks on these aspects are not feasible.

\begin{remark}
We assume without loss of generality that the message $x$ lives in $\mathbb{Z}_4^{k_1}\times \mathbb{Z}_2^{k_2}$. Indeed, we can transform any binary string $\bar{x} \in \mathbb{F}_2^{2k_1+k_2}$ into this form by an invertible map before the encryption, and use the inverse map after the decryption.
\end{remark}

\subsection{Quaternary McEliece}\label{mcelieceZ4}
Let $G$ be a $(k_1+k_2) \times n$ generator matrix of $\mathcal{C}$ and choose an  $n \times n$ permutation matrix $P$, this matrix has no further conditions, since the change of columns does not affect the $\mathbb{Z}_2$-part of the message, whereas for the $(k_1+k_2) \times (k_1+k_2)$ invertible matrix $S$, we need further conditions: in the classical case over finite fields, $S$ is just a change of basis, but in the $\mathbb{Z}_4$ case,  changing the rows of the generator matrix affects the position of $\mathbb{Z}_2$-part of the message. Since such a change  hinders the constructor of the cryptosystem to tell where the $\mathbb{Z}_2$-part of the message should be  taken, we will restrict the choice of invertible matrices to the following form: let $S_1$ and $S_2$ be  $k_1 \times k_1$, respectively $k_2 \times k_2$  invertible matrices over $\mathbb{Z}_4$, then $S$ is given by
\begin{equation*}
S= \begin{pmatrix}
S_1 & 0  \\
0 & S_2  
\end{pmatrix}.
\end{equation*} 
Compute $G' = SGP$ and publish $(k_1, k_2, G', t)$.

For the encryption, let $x =(x_1, x_2)$, with $x_1 \in \mathbb{Z}_4^{k_1}$ and $x_2 \in \mathbb{Z}_2^{k_2}$  be the message and choose an error vector $e \in \mathbb{Z}_4^{n}$ of Lee weight $\text{wt}_L(e) \leq t$. The cipher is computed as 
$$ y= xG' + e.$$

For the decryption one computes 
$$yP^{-1}= xSG +eP^{-1}.$$
Since $\text{wt}_L(eP^{-1}) \leq t$ and $SG$ generates the same code as $G$ we can use the decoding algorithm of the code to recover $xS$ and hence the message $x$.

\begin{remark}
To outgo a chosen ciphertext attack (CCA) one can multiply a new permutation matrix to the public generator matrix at each new instantiation of the system, analogously to the classical McEliece system \cite{attackinganddefending}.
\end{remark}

\subsection{Quaternary Niederreiter}\label{Nied}

The quaternary version of the Niederreiter cryptosystem is done in a similar way by using the parity check matrix $H$ and by computing its syndromes for encryption. Since there is no restriction on the message space in the Niederreiter version, there will be no conditions needed on the permutation matrix and on the invertible matrix.

Again, one chooses a quaternary code $\mathcal{C}$ of length $n$ and type $h=4^{k_1}2^{k_2}$, which has an efficient decoding algorithm and is able to correct up to $t$ errors. 

Let $H$ be a $(n-k_1) \times n$ parity matrix of $\mathcal{C}$, choose an invertible $(n-k_1) \times (n-k_1)$ matrix $S$, i.e. $\det(S) \in  \mathbb{Z}_4^{\times}$ and an $n \times n$ permutation matrix $P$.\footnote{Also here we can choose a new $P$ every time to prevent a CCA.}
Compute $H' = S^{-1}HP $ and publish $(k_1, k_2, H', t)$.

For the encryption, let $x \in \mathbb{Z}_4^n$ be the message of Lee weight $\text{wt}_L(x) \leq t$. The cipher is computed as 
$$ y^\top=  H'x^{\top}.$$
For the decryption one computes 
$$Sy^\top= HP x^{\top}.$$
Since $\text{wt}_L(P x^{\top}) \leq t$  we can use the  decoding algorithm of the code to recover $P x^{\top}$ and hence the message $x$.

       \subsection{Key size}\label{sec6}

 To determine the key size we need to count the number of non-prescribed entries of the public generator matrix. For this we assume that the generator matrix is published in quaternary systematic form as in \eqref{gen}. 
 
 This allows us to compute the size  of the generator matrix in the form \eqref{gen}, or equivalently the size of the parity check matrix in the form \eqref{parity}. 
 \begin{theorem}
 The size of the public key, given by the non-prescribed parts of either the generator matrix \eqref{gen} or the parity check matrix \eqref{parity}, is
\begin{equation*}
2(n-k_1-k_2)k_1+(n-k_1-k_2)k_2+k_1k_2=k_1 k_2 +(2k_1+k_2)(n-k_1-k_2)
\end{equation*}
    bits.
\end{theorem}

 In the following we study the key sizes of the proposed cryptographic scheme in Section~\ref{applications to mceliece over Z4} 
 with respect to a given security level against  Algorithm \ref{sternnew} provided in Section \ref{subsecstern}. 
 
Two of the most studied families of $\mathbb{Z}_4$-linear codes are Kerdock  and Preparata codes \cite{Ker}. Because of their small minimum distance Preparata codes are not useful for our cryptosystems. 
Even though Kerdock codes  over  $\mathbb{Z}_4$ satisfy all the conditions needed for the quaternary version of the McEliece cryptosystem, they seem to be a bad choice for key size reasons: while the key size of the cryptosystems doubles going from code length $n= 2^m$ to $2^{m+1}$, the security level only increases by 3 bits.  

For now we leave it as an open problem to find suitable codes for the use in a Lee-metric public key cryptosystem, but we remark that many constructions of Lee codes are known, e.g., \cite{cycliclee,eimearlee,leedense,Ker,goethalslee,perfectlee,leeqr,leebch2, leealternant,leebch}. For the remainder of this paper we will use only theoretical parameters, to illustrate how using the Lee metric could potentially decrease the key sizes in a McEliece or Niederreiter type cryptosystem. 
 We consider quaternary codes achieving the Gilbert-Varshamov bound in the Lee metric, i.e., codes of length $n$ and Lee weight $d$ whose cardinality is at least 
$$ \frac{4^n}{(\sum_{j=0}^{d-1} \binom{2n}{j}-1)3+1}.$$

\begin{example}
As a first example we examine codes of length $n=150$ and minimum Lee distance $d=81$, i.e., we can set $t=40$. The Gilbert-Varshamov bound tells us that such codes with $\mathbb Z_4$-dimension $26=k_1+k_2/2$ exist. We now vary $k_1$ from $1$ to $25$, with $k_2=2(26-k_1)$. Furthermore, we set $m_1=\lceil(k_1+k_2)/2\rceil, 
    m_2=\lfloor(k_1+k_2)/2\rfloor $ and we optimize on the size $\ell$ of the zero-window and the number $2v$ of errors in the information set.
With these parameters we get the following key sizes and security levels in bits, see Table \ref{tab1}.

\begin{table}[ht!]
\begin{tabular}{|c|c|c|c|c|c|c|c|c|c|c|c|}
\hline
$k_1$ & 1 & 2& 3& 4& \dots & 18 & 19& \dots & 24 &25 \\\hline
best $\ell$  & 1& 1&1&1& \dots&1&1&\dots &1&2 \\\hline
best $v$ & 4&4&4&4& \dots&3&3&\dots&3&3\\\hline
key size & 5198 & 5296 & 5390 & 5480 & \dots & 6110 & 6160 & \dots & 6440 & 6446\\\hline
security level  & 31 &31 &31 & 30 &  \dots & 27 & 27 & \dots & 28 & 28 \\\hline
\end{tabular}
\caption{Key sizes and security levels (both in bits) for GV-codes over $\mathbb Z_4$ with $n=150$ and $d=81$.}
\label{tab1}
\end{table}
Note that the above security levels were computed using Algorithm \ref{sternnew}, which always outperforms Algorithm \ref{leebrickellnew} on a classical computer. 
For comparison, a binary code of length $2n=300$, dimension $k=26$ and minimum Hamming distance $d=81$ gives a key size of $7124$ and a security level of $27$ bits with the binary version of Stern's algorithm. Hence, depending on $k_1$ we get at least the same security level with a key size improvement of around $10-28\%$, when using Lee codes over $\mathbb Z_4$.\footnote{
For further comparison, binary codes achieving the Gilbert-Varshamov bound of length $2n=300$ and  
dimension $k=26$ have minimum distance $d=102$ and achieve a security level of $28$ bits with a key size of $7124$ bits.}
\end{example}

In the next example we find codes that theoretically achieve a security level of $128$ bits, against the adaptation of the collision ISD algorithm.

\begin{example}
Given the  relative distance $d/n =0.2$, using an optimization on the size of $v$ and $\ell$ we search for a minimal code length $n$, such that a $k_1$ exists with which the security level of 128 bits is reached. We get $n=425, k=229, k_1=33, k_2=392, t= 42, v=21, \ell=0 $ and key size of 12936 bits.

\end{example}

\begin{remark} The previously obtained theoretical values give much smaller public keys than the classical McEliece system with binary Goppa codes achieves. For this note that for the security level of 128 bits, the proposed parameters for the McEliece system using Goppa codes by Bernstein \textit{et al.} in \cite{attackinganddefending}  are $n= 2960, k = 2288$, which gives a key size of $1537536$ bits. In fact the theoretical key sizes presented here are within the range of quasi-cyclic MDPC codes, which were submitted to NIST for post-quantum code-based cryptosystem (from 10 to 37 kilobits, see \url{https://csrc.nist.gov/projects/post-quantum-cryptography/round-2-submissions}). 
\end{remark}

\section{Generalization from $\mathbb Z_4$ to $\mathbb Z_{p^s}$}\label{generalize}

In this section we give the general idea of how to generalize the ISD algorithm and the two code-based cryptosystems to any Galois ring $\mathbb Z_{p^s}$ for any prime $p$ and $s\in \mathbb N$. 
Recall that the Lee weight of $x \in \mathbb{Z}_{p^s}^n$ is given by $$\text{wt}_L(x) = \sum_{i=0}^n \min\{x_i,p^s-x_i\}.$$
The main modification is in the systematic form of the generator and parity check matrix of the code. In general, a linear code over $\mathbb Z_{p^s}$ has a (column permuted) generator matrix of the form
\begin{equation}\label{gengen}
G = \begin{pmatrix}
\Id_{k_1} & A_{1,2} & A_{1,3} & \dots & A_{1,s} & A_{1,s+1} \\
0 & p\Id_{k_2} & pA_{2,3}& \dots &pA_{2,s} & pA_{2,s+1}\\
0 & 0 & p^2\Id_{k_3}& \dots &p^2A_{3,s} & p^2A_{3,s+1}\\
\vdots & \vdots &\ddots&&\vdots & \vdots\\
0&0&0 & \dots& p^{s-1} \Id_{k_s}& p^{s-1} A_{s,s+1}
\end{pmatrix},
\end{equation}
and a (column permuted) parity check matrix of the form
\begin{equation}\label{parpar}
H = \begin{pmatrix}
B_{1,1} & B_{1,2} & \dots & B_{1,s-1} & B_{1,s} & \Id_{n-K} \\
pB_{2,1} & pB_{2,2} &  \dots &pB_{2,s-1} & p\Id_{k_s} & 0\\
p^2B_{3,1} & p^2B_{3,2} &  \dots &p^2\Id_{k_{s-1}} & 0 & 0\\
\vdots & \vdots &\ddots&&\vdots & \vdots\\
p^{s-1}B_{s,1}& p^{s-1}\Id_{k_2}&\dots& 0& 0 & 0
\end{pmatrix},
\end{equation}
where $K= \sum_{i=1}^s k_i$ and the matrices live in
\begin{eqnarray*}
\text{for} \ j\leq s: A_{i,j} \in \mathbb{Z}_{p^{s+1-i}}^{k_i \times k_j}, \quad  \text{and} \quad A_{i,s+1} \in \mathbb{Z}_{p^{s+1-i}}^{k_i \times (n-K)}\\
\text{for} \  i>1: B_{i,j} \in \mathbb{Z}_{p^{s+1-i}}^{k_{s-j+2} \times k_j}, \quad  \text{and} \quad B_{1,j} \in \mathbb{Z}_{p^{s+1-i}}^{(n-K) \times k_j}.
\end{eqnarray*}
We say that such a code has type $(p^s)^{k_1}(p^{s-1})^{k_2}\dots p^{k_s}$. This is also the cardinality of the code, and the uniquely encodable messages are of the form $(m_1,m_2,\dots,m_s)\in \mathbb Z_{p^s}^{k_1}\times\mathbb Z_{p^{s-1}}^{k_2}\times \dots\times \mathbb Z_p^{k_s}$. If our code has length $n$, an information set is a set $I\subseteq \{1,\dots,n\}$ of size $K$ such that $|\mathcal C_I |=|\mathcal C|$.

\subsection{Information set decoding over $\mathbb Z_{p^s}$}

One can set up an ISD algorithm analogously to Algorithm \ref{sternnew}. Instead of three conditions on $e_1$ and $e_2$ we then get $s+1$ conditions on $e_i$ for $i \in \{1, \ldots s\}$, where $e_s$ is only part of one condition. For $e_i$ with $i \in \{1, \ldots s-1\}$ to  satisfy the $s$ conditions, that are not involving $e_s$, one needs to compute similar sets $S_i$ for $i \in \{1, \ldots, s-1\}$ consisting of  $s$  tuples, find the collisions between them, and lastly choose $e_s$ satisfying the remaining condition. In the appendix we exemplify the described algorithm over $\mathbb{Z}_8$. Note that this ISD algorithm is different to the Lee metric ISD algorithm proposed in \cite{LEE2}\footnote{\cite{LEE2} appeared as a follow-up of this work.}, as there the systematic form is considered to be \begin{equation*}
H= \begin{pmatrix}
A & \text{Id}_{n-K} \\
pB & 0
\end{pmatrix},
\end{equation*} with $A\in \mathbb{Z}_{p^s}^{(n-K) \times K}$ and $B \in \mathbb{Z}_{p^{s-1}}^{(K-k_1)\times K}.$ This choice of systematic form clearly makes the ISD algorithm easier to understand, but does not take into account the particular form of the parity check matrix. We leave it as an open problem to compare the benefits of the two different algorithms.

\subsection{Code-based cryptosystems over $\mathbb Z_{p^s}$}

The Niederreiter system does not need a modification of \ref{Nied} to be used over any $\mathbb Z_{p^s}$. 
In the McEliece cryptosystem \ref{mcelieceZ4} one just needs to make sure that the invertible matrix  $S$ has the correct block diagonal structure to prevent mixing the subcodes that live in different subrings. The rest stays the same.

For the size of the public key note that you can either publish the generator or the parity check matrix. However, it turns out that both ways give you the same key size. The key size in bits related to the generator matrix $G$ as in \eqref{gengen} (or equivalently related to the parity check matrix $H$ as in \eqref{parpar}) is 
$$\sum_{i=1}^s (\sum_{j=i+1}^s k_i k_j \log_2 (p^{s+1-i}) + k_i (n-K)\log_2 (p^{s+1-i})) $$
$$= \sum_{i=1}^s k_i \log_2 (p^{s+1-i}) \sum_{j=i+1}^s  (k_j +  n-K) .$$

\section{Conclusion}\label{conclusion}

The change from the Hamming metric to the rank metric  has recently received a lot of attention in the code-based cryptography community, since the key sizes are very promising. Following this idea, we propose the change to the Lee metric and the ring-linear codes related to this metric.
In this paper we built the framework for the use of quaternary codes in code-based cryptography by generalizing Lee-Brickell's and Stern's ISD algorithm to $\mathbb{Z}_4$. 
This paper also gives the general form of the quaternary version of the McEliece and the Niederreiter cryptosystem.

Here we provide some questions, which might lead to interesting applications and further understanding of  ring-linear codes and the Lee metric from a cryptographic point of view. 
Even though we restricted the focus in this paper to the case $\mathbb{Z}_4$, and explained shortly how to generalize this to $\mathbb{Z}_{p^s}$, it is possible and it might be interesting to generalize this to $\mathbb{Z}_m$, for any $m$.
It is possible that some of the ISD algorithms have a structure that correlates better to the Lee metric, hence there might be other ISD algorithms which should be generalized to $\mathbb{Z}_4$. And the most important question in order to have an application in cryptography is: which codes might be used for the  quaternary version of the McEliece cryptosystem, such that the conditions for the cryptosystem are satisfied and the key size is reasonable?

\section*{Acknowledgments}
We would like to thank Karan Khathuria for fruitful discussions and technical support.

\bibliographystyle{plain}

\bibliography{biblio}

\begin{thebibliography}{10}

\bibitem{ringlin}
Edward~F. Assmus and Harold~F. Mattson.
\newblock Error-correcting codes: an axiomatic approach.
\newblock {\em Information and Control}, 6(4):315--330, 1963.

\bibitem{BJMM}
Anja Becker, Antoine Joux, Alexander May, and Alexander Meurer.
\newblock Decoding random binary linear codes in $2^{n/20}$: How 1+ 1= 0
  improves information set decoding.
\newblock In {\em Annual International Conference on the Theory and
  Applications of Cryptographic Techniques}, pages 520--536. Springer, 2012.

\bibitem{berlekampbook}
Elwyn Berlekamp.
\newblock {\em Algebraic coding theory}.
\newblock World Scientific, 1968.

\bibitem{GroverISD}
Daniel~J Bernstein.
\newblock {G}rover vs. {McE}liece.
\newblock In {\em International Workshop on Post-Quantum Cryptography}, pages
  73--80. Springer, 2010.

\bibitem{attackinganddefending}
Daniel~J Bernstein, Tanja Lange, and Christiane Peters.
\newblock Attacking and defending the {McE}liece cryptosystem.
\newblock In {\em International Workshop on Post-Quantum Cryptography}, pages
  31--46. Springer, 2008.

\bibitem{ballcoll}
Daniel~J. Bernstein, Tanja Lange, and Christiane Peters.
\newblock Smaller decoding exponents: ball-collision decoding.
\newblock In {\em Annual Cryptology Conference}, pages 743--760. Springer,
  2011.

\bibitem{cycliclee}
Thomas Blackford.
\newblock Cyclic codes over {$\mathbb{Z}_4$} of oddly even length.
\newblock {\em Discrete Applied Mathematics}, 128(1):27--46, 2003.

\bibitem{blake2}
Ian~F. Blake.
\newblock Codes over certain rings.
\newblock {\em Information and Control}, 20(4):396--404, 1972.

\bibitem{blake1}
Ian~F. Blake.
\newblock Codes over integer residue rings.
\newblock {\em Information and Control}, 29(4):295--300, 1975.

\bibitem{eimearlee}
Eimear Byrne.
\newblock Decoding a class of {L}ee metric codes over a galois ring.
\newblock {\em IEEE Transactions on Information Theory}, 48(4):966--975, 2002.

\bibitem{chabanne}
Anne Canteaut and Herv{\'e} Chabanne.
\newblock {\em A further improvement of the work factor in an attempt at
  breaking {McE}liece's cryptosystem}.
\newblock PhD thesis, INRIA, 1994.

\bibitem{canteaut}
Anne Canteaut and Florent Chabaud.
\newblock A new algorithm for finding minimum-weight words in a linear code:
  application to {McE}liece's cryptosystem and to narrow-sense {BCH} codes of
  length 511.
\newblock {\em IEEE Transactions on Information Theory}, 44(1):367--378, 1998.

\bibitem{canteautsendrier}
Anne Canteaut and Nicolas Sendrier.
\newblock Cryptanalysis of the original {McE}liece cryptosystem.
\newblock In {\em International Conference on the Theory and Application of
  Cryptology and Information Security}, pages 187--199. Springer, 1998.

\bibitem{chabaud}
Florent Chabaud.
\newblock Asymptotic analysis of probabilistic algorithms for finding short
  codewords.
\newblock In {\em Eurocode’92}, pages 175--183. Springer, 1993.

\bibitem{coffey}
John~T. Coffey and Rodney~M. Goodman.
\newblock The complexity of information set decoding.
\newblock {\em IEEE Transactions on Information Theory}, 36(5):1031--1037,
  1990.

\bibitem{dumer}
Il'ya~Isaakovich Dumer.
\newblock Two decoding algorithms for linear codes.
\newblock {\em Problemy Peredachi Informatsii}, 25(1):24--32, 1989.

\bibitem{leedense}
Tuvi Etzion, Alexander Vardy, and Eitan Yaakobi.
\newblock Dense error-correcting codes in the {L}ee metric.
\newblock In {\em 2010 IEEE Information Theory Workshop}, pages 1--5. IEEE,
  2010.

\bibitem{sendrier}
Matthieu Finiasz and Nicolas Sendrier.
\newblock Security bounds for the design of code-based cryptosystems.
\newblock In {\em International Conference on the Theory and Application of
  Cryptology and Information Security}, pages 88--105. Springer, 2009.

\bibitem{gabidulinmetrics}
Ernst Gabidulin.
\newblock A brief survey of metrics in coding theory.
\newblock {\em Mathematics of Distances and Applications}, 66, 2012.

\bibitem{gref}
Marcus Greferath.
\newblock An introduction to ring-linear coding theory.
\newblock In {\em Gr{\"o}bner Bases, Coding, and Cryptography}, pages 219--238.
  Springer, 2009.

\bibitem{Ker}
Roger~A. Hammons, Vijay~P. Kumar, Robert~A. Calderbank, Neil Sloane, and
  Patrick Sol{\'e}.
\newblock The $\mathbb{Z}_4$-linearity of {K}erdock, {P}reparata, {G}oethals,
  and related codes.
\newblock {\em IEEE Transactions on Information Theory}, 40(2):301--319, 1994.

\bibitem{goethalslee}
Tor Helleseth and Victor Zinoviev.
\newblock On {$\mathbb{Z}_4$}-linear {G}oethals codes and {K}loosterman sums.
\newblock {\em Designs, Codes and Cryptography}, 17(1-3):269--288, 1999.

\bibitem{hirose}
Shoichi Hirose.
\newblock May-{O}zerov algorithm for nearest-neighbor problem over
  $\mathbb{F}_q$ and its application to information set decoding.
\newblock In {\em International Conference for Information Technology and
  Communications}, pages 115--126. Springer, 2016.

\bibitem{ballcollFq}
Carmelo Interlando, Karan Khathuria, Nicole Rohrer, Joachim Rosenthal, and
  Violetta Weger.
\newblock Generalization of the ball-collision algorithm.
\newblock {\em arXiv preprint arXiv:1812.10955}, 2018.

\bibitem{perfectlee}
Denis~S Krotov.
\newblock {$\mathbb{Z}_4$}-linear {H}adamard and extended perfect codes.
\newblock {\em Electronic Notes in Discrete Mathematics}, 6:107--112, 2001.

\bibitem{kruk}
Evgenii~Avramovich Kruk.
\newblock Decoding complexity bound for linear block codes.
\newblock {\em Problemy Peredachi Informatsii}, 25(3):103--107, 1989.

\bibitem{leebrickell}
Pil~Joong Lee and Ernest~F. Brickell.
\newblock An observation on the security of {McE}liece’s public-key
  cryptosystem.
\newblock In {\em Workshop on the Theory and Application of of Cryptographic
  Techniques}, pages 275--280. Springer, 1988.

\bibitem{leon}
Jeffrey~S. Leon.
\newblock A probabilistic algorithm for computing minimum weights of large
  error-correcting codes.
\newblock {\em IEEE Transactions on Information Theory}, 34(5):1354--1359,
  1988.

\bibitem{MMT}
Alexander May, Alexander Meurer, and Enrico Thomae.
\newblock Decoding random linear codes in $\mathcal{O} (2^{0.054 n})$.
\newblock In {\em International Conference on the Theory and Application of
  Cryptology and Information Security}, pages 107--124. Springer, 2011.

\bibitem{MO}
Alexander May and Ilya Ozerov.
\newblock On computing nearest neighbors with applications to decoding of
  binary linear codes.
\newblock In {\em Annual International Conference on the Theory and
  Applications of Cryptographic Techniques}, pages 203--228. Springer, 2015.

\bibitem{mceliece}
Robert~J. McEliece.
\newblock A {P}ublic-{K}ey {C}ryptosystem {B}ased on {A}lgebraic {C}oding
  {T}heory.
\newblock Technical report, DSN Progress report, Jet Propulsion Laboratory,
  Pasadena, 1978.

\bibitem{meurer}
Alexander Meurer.
\newblock {\em A coding-theoretic approach to cryptanalysis}.
\newblock PhD thesis, Ruhr University Bochum, 2012.

\bibitem{nechaev}
Alexander~A. Nechaev.
\newblock Kerdock code in a cyclic form.
\newblock {\em Discrete Mathematics and Applications}, 1(4):365--384, 1991.

\bibitem{Niebuhr}
Robert Niebuhr, Edoardo Persichetti, Pierre-Louis Cayrel, Stanislav Bulygin,
  and Johannes Buchmann.
\newblock On lower bounds for information set decoding over $\mathbb{F}_q$ and
  on the effect of partial knowledge.
\newblock {\em International journal of information and Coding Theory},
  4(1):47--78, 2017.

\bibitem{peters}
Christiane Peters.
\newblock Information-set decoding for linear codes over $\mathbb{F}_q$.
\newblock In {\em International Workshop on Post-Quantum Cryptography}, pages
  81--94. Springer, 2010.

\bibitem{leeqr}
Vera~S Pless and Zhongqiang Qian.
\newblock Cyclic codes and quadratic residue codes over {$\mathbb{Z}_4$}.
\newblock {\em IEEE Transactions on Information Theory}, 42(5):1594--1600,
  1996.

\bibitem{prange}
Eugene Prange.
\newblock The use of information sets in decoding cyclic codes.
\newblock {\em IRE Transactions on Information Theory}, pages 5--9, 1962.

\bibitem{leebch2}
Ron~M Roth and Paul~H Siegel.
\newblock Lee-metric {BCH} codes and their application to constrained and
  partial-response channels.
\newblock {\em IEEE Transactions on Information Theory}, 40(4):1083--1096,
  1994.

\bibitem{saty}
Chandra Satyanarayana.
\newblock Lee metric codes over integer residue rings (corresp.).
\newblock {\em IEEE Transactions on Information Theory}, 25(2):250--254, 1979.

\bibitem{shankar}
Priti Shankar.
\newblock On {BCH} codes over arbitrary integer rings (corresp.).
\newblock {\em IEEE Transactions on Information Theory}, 25(4):480--483, 1979.

\bibitem{spiegel}
Eugene Spiegel.
\newblock Codes over $\mathbb{Z}_m$.
\newblock {\em Information and control}, 35(1):48--51, 1977.

\bibitem{stern}
Jacques Stern.
\newblock A method for finding codewords of small weight.
\newblock In {\em International Colloquium on Coding Theory and Applications},
  pages 106--113. Springer, 1988.

\bibitem{leealternant}
Ido Tal and Ronny~M Roth.
\newblock On list decoding of alternant codes in the {H}amming and {L}ee
  metrics.
\newblock In {\em IEEE International Symposium on Information Theory}, pages
  364--364, 2003.

\bibitem{ringlincrypto}
Horacio Tapia-Recillas.
\newblock A secret sharing scheme from a chain ring linear code.
\newblock {\em Congressus Numerantium}, 186:33, 2007.

\bibitem{vantilborg}
Johan van Tilburg.
\newblock On the {McE}liece public-key cryptosystem.
\newblock In {\em Conference on the Theory and Application of Cryptography},
  pages 119--131. Springer, 1988.

\bibitem{LEE2}
Violetta Weger, Massimo Battaglioni, Paolo Santini, Franco Chiaraluce, Marco
  Baldi, and Edoardo Persichetti.
\newblock Information set decoding of {L}ee-metric codes over finite rings.
\newblock {\em arXiv preprint arXiv:2001.08425}, 2020.

\bibitem{leebch}
Yingquan Wu and Christoforos~N Hadjicostis.
\newblock Decoding algorithm and architecture for {BCH} codes under the lee
  metric.
\newblock {\em IEEE transactions on communications}, 56(12):2050--2059, 2008.

\end{thebibliography}

\section*{Appendix}
Here we describe the ISD algorithm from Section \ref{generalize} over $\mathbb{Z}_8$. We consider a code 
$\mathcal{C}\subseteq \mathbb{Z}_8^n$ of type $\mid \mathcal{C} \mid = 8^{k_1}4^{k_2}2^{k_3}$, and define $K:=k_1+k_2+k_3$.
For simplicity let us assume that the  information set is $I=\{1,\dots,K\}$ with $I_1=\{1, \ldots, k_1+k_2\}$, $I_2=\{k_1+k_2+1, \ldots, K\}$ and the zero window is $\{K+1, \ldots, K+\ell\}$.
We first bring the parity check matrix into systematic form, according to the chosen information set, and get:
  \begin{equation*}
 UHe^{\top} = \begin{pmatrix}
 A &  B & \Id_{\ell} & 0 \\
 C & D  & 0 & \Id_{n-K-\ell} \\
 2E & 2 \text{Id}_{k_3}  & 0 & 0 \\
 4F & 0 & 0 & 0
 \end{pmatrix}\begin{pmatrix}
 e_1^\top \\ e_2^\top \\ 0  \\ e_3^\top
 \end{pmatrix} = \begin{pmatrix}
 s_1^\top \\ s_2^\top \\ 2s_3^\top \\4 s_4^\top
 \end{pmatrix} =Us,
 \end{equation*}
 where $s_1 \in \mathbb{Z}_8^{\ell}, s_2 \in \mathbb{Z}_8^{n-K-\ell}, s_3 \in \mathbb{Z}_4^{k_3}, s_4 \in \mathbb{Z}_2^{k_2}$ and  $A \in \mathbb{Z}_8^{\ell \times (k_1+k_2)}$, $B \in \mathbb{Z}_8^{\ell \times k_3}$,  $C \in \mathbb{Z}_8^{(n-K-\ell)\times  (k_1+k_2)}, D \in \mathbb{Z}_8^{(n-K-\ell) \times k_3}, E \in \mathbb{Z}_4^{k_3 \times (k_1 + k_2)}, F \in \mathbb{Z}_2^{k_2 \times (k_1 +k_2)}$ and $e_1 \in \mathbb{Z}_8^{k_1+k_2}, e_2  \in \mathbb{Z}_8^{k_3}, e_3 \in \mathbb{Z}_8^{n-K-\ell}$. From this we get the conditions
 \begin{align*}
 e_1A^\top +e_2B^\top = & s_1  , \\
 e_1C^\top  +e_2D^\top +e_3 = & s_2 , \\
 2 e_1E^\top  +2e_2 = & 2s_3, \\
 4e_1F^\top = & 4s_4 .
 \end{align*}
  We will choose $e_1$ and $e_2$ disjoint both having Lee weight $v$ and $e_3$ having Lee weight $t-2v$. In order to satisfy the first, the third  and the fourth condition, which only depend on $e_1$ and $e_2$, we will check for a collision within the algorithm. The second condition will be satisfied by choosing  $e_3=   s_2- e_1C^\top-e_2D^\top$. The algorithm is provided in Algorithm \ref{sterngen8}.
 
\begin{algorithm}[h!]
\caption{Collision ISD (Stern's algorithm) over $\mathbb{Z}_8$}\label{sterngen8}
\begin{flushleft}
Input: The $(n-k_1)\times n$ parity check matrix $H$ over $\mathbb{Z}_8$, the syndrome $s\in\mathbb{Z}_8^{n-k_1}$ and the positive integers $v,  \ell \in \mathbb{Z}$, such that  $v \leq \min\{ 4(k_1+k_2), 4k_3\}$,  $2v\leq t$ and $t-2v \leq 4(n-K- \ell)$.\\ 
Output: $e\in\mathbb{Z}_8^n$ with $He^{\top}=s^\top$ and $wt_L(e)=t$.
\end{flushleft}
\begin{algorithmic}[1]
\State Choose a quaternary information set $I \subset \{1, \ldots, n\}$ of size $K= k_1+k_2+k_3$ with the corresponding subsets $I_1$ of size $(k_1+k_2)$ and $I_2$ of size $k_3$.
\State Choose a set $Z \subset \{1, \ldots, n\} \setminus I$, of size $\ell$ and define $J =\{1, \ldots ,n\} \setminus (I \cup Z)$. 
\State\label{find Unew2}  Find an invertible matrix $U\in\mathbb{Z}_8^{(n-k_1)\times(n-k_1)} $, such that  \begin{align*}
(UH)_{I_1} & = \begin{pmatrix}
A \\ C \\ 2E \\ 4F
\end{pmatrix}, \quad & (UH)_{I_2} = \begin{pmatrix}
B \\ D \\ 2\Id_{k_3} \\ 0
\end{pmatrix}, \\ (UH)_{Z} 
 & = \begin{pmatrix}
\Id_{\ell} \\0\\ 0 \\ 0
\end{pmatrix}, \quad & (UH)_{J}  = \begin{pmatrix}
0 \\ \Id_{n-K-\ell} \\ 0 \\ 0
\end{pmatrix}, 
\end{align*}
 where  $A \in \mathbb{Z}_8^{\ell \times (k_1+k_2)}$, $B \in \mathbb{Z}_8^{\ell \times k_3}$,  $C \in \mathbb{Z}_8^{(n-K-\ell)\times  (k_1+k_2)}, D \in \mathbb{Z}_8^{(n-K-\ell) \times k_3}, E \in \mathbb{Z}_4^{k_3 \times (k_1 + k_2)}$ and  $F \in \mathbb{Z}_2^{k_2 \times (k_1 +k_2)}$.
\State Compute $Us^\top= \begin{pmatrix} s_1^\top \\ s_2^\top \\ 2s_3^\top \\ 4s_4^\top \end{pmatrix}$, where  $s_1 \in \mathbb{Z}_8^{\ell}, s_2 \in \mathbb{Z}_8^{n-K-\ell}, s_3 \in \mathbb{Z}_4^{k_3}$ and $s_4 \in \mathbb{Z}_2^{k_2}$.
\State Compute the following set \label{buildSgen}
	\begin{equation*}
	S= \{(Ae_1^\top ,2Ee_1^\top , 4Fe_1^\top , e_1)  | \  e_1 \in \mathbb{Z}_8^{k_1+k_2}, wt_L(e_1) = v \}.
	\end{equation*}
\State Compute the following set \label{buildTgen}
\begin{equation*}
T=\{(s_1^\top -Be_2^\top ,  2s_3^\top -2e_2^\top , 4s_4^\top, e_2) | \   e_2 \in \mathbb{Z}_8^{k_3}, wt_L(e_2)=v \}.
\end{equation*}
	\For{$(a,  b, c, e_1)\in S$}
	\For{$(a, b, c, e_2)\in T$} \label{collTnewgen}
	\If{$wt_L( s_2^\top- Ce_1^\top-De_2^\top )=t-2v$}
	\State  Output: $e_{I_1}= e_1, e_{I_2}=e_2, e_Z=0$ and $e_J= s_2- e_1C^\top-e_2D^\top$.
				\EndIf
		\EndFor
		\EndFor
		\State  	Start over with Step 1 and a new selection of $I$.
\end{algorithmic}
 
\end{algorithm}

\end{document}